\RequirePackage{mathtools}

\documentclass[envcountsame]{llncs}

\usepackage[T1]{fontenc}
\usepackage{microtype}
\usepackage[misc,geometry]{ifsym} 

\usepackage{amssymb}
\let\oldendproof\endproof
\renewcommand\endproof{~\hfill$\qed$\oldendproof}

\usepackage{subcaption}
\usepackage{complexity}
\usepackage{xspace}
\usepackage{csquotes}
\usepackage{thm-restate}

\usepackage[svgnames]{xcolor}
\usepackage{hyperref}
\hypersetup{
    breaklinks=true,
    colorlinks=true,
    linkcolor=DarkBlue,
    urlcolor=MediumVioletRed,
    citecolor=Green
}

\DeclareCaptionFormat{myformat}{{\bfseries #1#2}#3}
\renewcommand{\R}{\ensuremath{\mathbb{R}}\xspace}
\renewcommand{\H}{\ensuremath{\mathbb{H}}\xspace}
\newcommand{\decproblemname}[1]{{\normalfont\textsc{#1}}\xspace}
\newcommand{\ETR}{\decproblemname{ETR}}
\newcommand{\Stretchability}{\decproblemname{Stretchability}}
\newcommand{\SimpleStretchability}{\decproblemname{SimpleStretchability}}
\newcommand{\ER}{\ensuremath{\exists\R}\xspace}
\newcommand{\calA}{\mathcal{A}\xspace}
\newcommand{\calD}{\mathcal{D}\xspace}
\newcommand{\calR}{\mathcal{R}\xspace}
\newcommand{\Poincare}{Poincar{\'{e}}\xspace}
\newcommand{\eps}{\varepsilon}
\newcommand{\abs}[1]{\lvert #1 \rvert}

\spnewtheorem{openproblem}{Open Problem}{\bfseries}{\itshape}

\begin{document}

\title{On the Complexity of Lombardi Graph Drawing}

\author{Paul Jungeblut\inst{1}\orcidID{0000-0001-8241-2102} (\Letter)}
\authorrunning{P. Jungeblut}
\institute{Karlsruhe Institute of Technology, Karlsruhe, Germany, \email{paul.jungeblut@kit.edu}}

\maketitle

\begin{abstract}
    In a \emph{Lombardi drawing} of a graph the vertices are drawn as points and the edges are drawn as circular arcs connecting their respective endpoints.
    Additionally, all vertices have perfect angular resolution, i.e., all angles incident to a vertex~$v$ have size~$2\pi/\deg(v)$.
    We prove that it is \ER-complete to determine whether a given graph admits a Lombardi drawing respecting a fixed cyclic ordering of the incident edges around each vertex.
    In particular, this implies \NP-hardness.
    While most previous work studied the (non-)existence of Lombardi drawings for different graph classes, our result is the first on the computational complexity of finding Lombardi drawings of general graphs.

\end{abstract}

\section{Introduction}

Inspired by the work of American artist Mark Lombardi~\cite{Hobbs2003_GlobalNetworks}, a \emph{Lombardi drawing} of a given graph~$G$ draws vertices as points and edges as circular arcs or line segments connecting their endpoints.
Further, each vertex~$v$ has \emph{perfect angular resolution}, meaning that all angles between edges incident to~$v$ have size~$2\pi/\deg(v)$.
Notably, planarity is not required (even for planar graphs) and the crossing angle at intersections may be arbitrary.
See Figure~\ref{fig:lombardi_examples} for Lombardi drawings of three well-known graphs.

\begin{figure}
    \begin{subfigure}[t]{0.3\textwidth}
        \centering
        \includegraphics[page=1,scale=0.7]{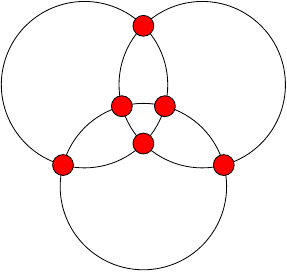}
        \caption{Octahedron Graph}
    \end{subfigure}
    \hfill
    \begin{subfigure}[t]{0.3\textwidth}
        \centering
        \includegraphics[page=2,scale=0.7]{figures/lombardi-examples.pdf}
        \caption{Petersen Graph}
    \end{subfigure}
    \hfill
    \begin{subfigure}[t]{0.35\textwidth}
        \centering
        \includegraphics[page=3,scale=0.7]{figures/lombardi-examples.pdf}
        \caption{Gr{\"{o}}tzsch Graph}
    \end{subfigure}
    \caption{Lombardi drawings created with the \emph{Lombardi Spirograph} from~\cite{Duncan2012_Lombardi}.}
    \label{fig:lombardi_examples}
\end{figure}

Introduced by Duncan, Eppstein, Goodrich, Kobourov and N{\"{o}}llenburg over ten years ago~\cite{Duncan2012_Lombardi}, Lombardi drawings have received a lot of attention in the graph drawing community, see the related work in Section~\ref{sec:related_work} below.
While most literature focuses on the construction of Lombardi drawings for different graph classes, the computational complexity to decide whether a Lombardi drawing exists remains largely unknown.
To the best of our knowledge, \NP-completeness is only known for certain regular graphs under the additional requirement that all vertices must lie on a common circle~\cite{Duncan2012_Lombardi}.
No lower or upper bounds on the complexity for general graphs are known (allowing arbitrary vertex placement).

In this paper we consider the case that the graph~$G$ comes with a fixed \emph{rotation system}~$\calR$, i.e., a cyclic ordering of the incident edges around each vertex.
Our main result is to determine the exact computational complexity of deciding whether~$G$ admits a Lombardi drawing respecting~$\calR$:

\begin{theorem}
    \label{thm:lombardi_er_complete}
    Given a graph~$G$ with a rotation system~$\calR$, it is \ER-complete to decide whether~$G$ admits a Lombardi drawing respecting~$\calR$.
\end{theorem}

The complexity class \ER contains all problems that can be reduced to solving a system of polynomial equations and inequalities, see Section~\ref{sec:etr} for a formal definition.
Since~$\NP \subseteq \ER$, our result also implies \NP-hardness.

\medskip
Previous work frequently utilizes hyperbolic geometry to construct Lombardi drawings~\cite{Duncan2012_Lombardi,Eppstein2014_SoapBubbles,Eppstein2021_HyperbolicGD}, the reason being that straight line segments in the hyperbolic plane~$\H^2$ can be visualized by circular arcs in the Euclidean plane~$\R^2$ (with the same crossing angles).
We take a similar approach:
A key ingredient of our \ER-hardness reduction is a recent observation by Bieker, Bl{\"{a}}sius, Dohse and Jungeblut~\cite{Bieker2023_HyperbolicER} stating that a simple pseudoline arrangements is stretchable in the Euclidean plane~$\R^2$ if and only if it is stretchable in the hyperbolic plane~$\H^2$ (see Sections~\ref{sec:hyperbolic} and~\ref{sec:stretchability} for the necessary definitions).
Their result allows us on the one hand to construct Lombardi drawings from hyperbolic line arrangements, and on the other hand to prove that sometimes no Lombardi drawing can exist.

\subsection{Related Work}
\label{sec:related_work}

Lombardi drawings were introduced by Duncan, Eppstein, Goodrich, Kobourov and N{\"{o}}llenburg~\cite{Duncan2012_Lombardi}, motivated by the network visualizations of Mark Lombardi~\cite{Hobbs2003_GlobalNetworks}.
While not all graphs admit Lombardi drawings (with or without prescribing the rotation system)~\cite{Duncan2018_PlanarPolyArcLombardi,Duncan2012_Lombardi}, many graph classes always admit Lombardi drawings.
Among them are $2$-degenerate (and some $3$-degenerate) graphs~\cite{Duncan2012_Lombardi}, subclasses of $4$-regular graphs~\cite{Duncan2012_Lombardi,Kindermann2019_Knots} and many classes of planar graphs that even admit planar Lombardi drawings.
These include trees~\cite{Duncan2013_Trees}, cactus graphs~\cite{Eppstein2023_Angles}, Halin graphs~\cite{Duncan2012_Lombardi,Eppstein2016_HalinRecognition}, subcubic graphs~\cite{Eppstein2014_SoapBubbles} and outerpaths~\cite{Duncan2018_PlanarPolyArcLombardi}.
However, many planar graphs do not admit planar Lombardi drawings in general~\cite{Duncan2018_PlanarPolyArcLombardi,Duncan2012_Lombardi,Eppstein2014_SoapBubbles,Eppstein2021_BipartiteSP,Kindermann2019_Knots}.

A user study confirmed that Lombardi drawings are considered more aesthetic than straight line drawings but do not increase the readability~\cite{Purchase2013_UserStudy}.

Many variants have been considered:
In a \emph{$k$-circular Lombardi drawing} all vertices lie on one of~$k$ concentric circles~\cite{Duncan2012_Lombardi}.
Slightly relaxing the perfect angular resolution condition leads to \emph{near Lombardi drawings}~\cite{Chernobelskiy2012_ForceLombardi,Kindermann2019_Knots}.
Lastly, edges in \emph{$k$\nobreakdash-Lombardi drawings} are drawn as the concatenation of up to~$k$ circular arcs~\cite{Duncan2018_PlanarPolyArcLombardi,Kindermann2019_Knots}.

Not much is known regarding the computational complexity of deciding whether a given graph admits a Lombardi drawing.
Proving that a graph class always admits a Lombardi drawing is usually done constructively (this is the case for all classes mentioned above).
In fact, these proofs lead to efficient algorithms (at least in a real RAM model of computation where square roots can be computed exactly).
On the other hand, for $d$-regular graphs with~$d \equiv 2 \mod 4$ it is \NP-complete to decide whether they have a $1$-circular Lombardi drawing~\cite{Duncan2012_Lombardi}.
Containment in \NP~might be surprising as this is in contrast to our main result showing \ER-hardness for general graphs and \enquote{classical} Lombardi drawings.
\NP-membership follows, because those graphs are yes-instances if and only if they are Hamiltonian.

\section{Preliminaries}

Let us recall the necessary geometric foundation for our reduction.

\subsection{Hyperbolic Geometry}
\label{sec:hyperbolic}

The hyperbolic plane~$\H^2$ is an example of a non-Euclidean geometry.
In many ways it behaves similar to the Euclidean plane~$\R^2$, e.g.\ two points define a unique line and we can measure distances and angles.

Formally, both~$\R^2$ and~$\H^2$ can be described by an axiomatic system (like the one from Hilbert for~$\R^2$~\cite{Hilbert1968_GrundlagenGeometrie}).
In fact, axiomatic systems for~$\R^2$ and~$\H^2$ are nearly identical, explaining the many similarities between~$\R^2$ and~$\H^2$.
Without going into the technical details, Hilbert's axiomatic system contains the so-called \emph{parallel postulate} stating that for any line~$\ell$ and point~$p$ not on~$\ell$ in~$\R^2$ there is at most one (indeed exactly one) line through~$p$ parallel to~$\ell$.
Negating this axiom turns Hilbert's axiomatic system for~$\R^2$ into one that defines~$\H^2$.

When working with the hyperbolic plane~$\H^2$ we usually avoid working with the axioms directly.
Instead we consider so-called \emph{models}, i.e., embeddings of~$\H^2$ into (in our case)~$\R^2$.
Several of these models are used in the literature.
Important for us is the \emph{\Poincare disk model}, see Figure~\ref{fig:models}, where~$\H^2$ is mapped to the interior of a unit disk~$D$ called the \emph{\Poincare disk}.
We omit how~$\H^2$ is mapped into~$D$ and instead focus on some useful properties:
\begin{itemize}
    \item Hyperbolic lines are mapped to either circular arcs orthogonal to~$D$ or diameters of~$D$.
    By a slight perturbation it is actually always possible to obtain a realization in the \Poincare disk in which each hyperbolic line is represented by a circular arc.
    \item The \Poincare disk model is \emph{conformal}, meaning that the angles in the hyperbolic plane equal the angles in a drawing inside the \Poincare disk~$D$.
    Conformality is crucial in our reduction to obtain perfect angular resolution.
\end{itemize}

\addtocounter{figure}{-1}
\begin{figure}[tb]
    \begin{minipage}[t]{0.58\textwidth}
        \begin{subfigure}[t]{\textwidth}
            \centering
            \includegraphics[page=2]{figures/models.pdf}
        \end{subfigure}
        \captionof{figure}{Hyperbolic lines in the \Poincare disk~$D$.}
        \label{fig:models}
    \end{minipage}
    \hfill
    \begin{minipage}[t]{0.38\textwidth}
        \begin{subfigure}[t]{0.98\textwidth}
            \centering
            \includegraphics[page=1]{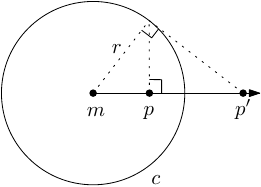}
        \end{subfigure}
        \captionof{figure}{
            Circle inversion.
        }
        \label{fig:circle_inversion}
    \end{minipage}
\end{figure}

\subsection{Complexity Class \ER}
\label{sec:etr}

Intuitively, the complexity class \ER contains all problems that can be formulated as a system of polynomial equations and inequalities.
Formally, it is defined to contain all problems that polynomial-time many-one reduce to the decision problem \ETR (short for \enquote{existential theory of the reals}) which is defined as follows:
The input of \ETR is a well-formed sentence~$\Phi$ in the existential fragment of the first-order theory of the reals, i.e., a sentence of the form
\[
    \Phi~\equiv~
    \exists X_1, \ldots, X_n \in \R :
    \varphi(X_1, \ldots, X_n)
    \text{,}
\]
where~$\varphi$ is a quantifier-free formula consisting of polynomial equations and inequalities with integer coefficients.
The task is to decide whether~$\Phi$ is true.
For example, $\exists X, Y \in \R : XY - 2X = 1 \land X + Y = 4$ is a yes-instance of \ETR because for~$(X,Y) = (1,3)$ both polynomial equations are satisfied.
On the other hand, $\exists X \in \R : X^2 < 0$ is a no-instance:
There is no real number~$X$ whose square is negative.
It is known that~$\NP \subseteq \ER \subseteq \PSPACE$ and both inclusions are conjectured to be strict~\cite{Canny1988_PSPACE,Schaefer2017_FixedPointsNash,Shor1991_Stretchability}.

Many problems from computational geometry and especially graph drawing have been shown to be \ER-complete.
Examples include RAC-drawability~\cite{Schaefer2021_RAC}, geometric $k$-planarity~\cite{Schaefer2021_KPlanarityFixed}, the recognition of many types of geometric intersection graphs~\cite{Bieker2023_HyperbolicER,Kratochvil1994_SegmentIntersectionGraphs,Schaefer2010_GeometryTopology}, simultaneous graph embedding~\cite{Cardinal2015_SimultaneousEmbedding,Kyncl2011_SimultaneousEmbedding,Schaefer2021_FixedK}, variants of the segment number~\cite{Okamoto2019_SegmentNumber}, as well as extending a partial planar straight line drawing of a planar graph inside a polygonal region to a drawing of the full graph~\cite{Lubiw2022_DrawingInPolygonialRegion}.

\subsection{Stretchability of Pseudolines}
\label{sec:stretchability}

A \emph{pseudoline arrangement}~$\calA = \{\ell_1, \ldots, \ell_n\}$ is a set of $x$-monotone curves in~$\R^2$ such that each pair of curves intersects at most once.
We say that~$\calA$ is \emph{simple} if each pair intersects exactly once and no three pseudolines in~$\calA$ intersect in a point.
See Figure~\ref{fig:stretchability_arrangement} for a simple pseudoline arrangement.

We always assume that the pseudolines are labeled such that a vertical line to the left of all intersections crosses~$\ell_i$ below~$\ell_j$ (for~$i, j \in \{1, \ldots, n\}$) if and only if~$i < j$.
There are several equivalent ways to describe the intersection pattern of a given pseudoline arrangement.
For us, a \emph{combinatorial description}~$\calD$ of~$\calA$ is a list of~$n$ lists, one for each pseudoline, listing the order of intersections along it from left to right.
For example, the list of intersections for~$\ell_1$ in Figure~\ref{fig:stretchability_arrangement} contains (in this order) $\ell_3$, $\ell_4$ and~$\ell_2$.

\begin{figure}[tb]
    \centering
    \begin{subfigure}[t]{0.3\textwidth}
        \centering
        \includegraphics[page=1]{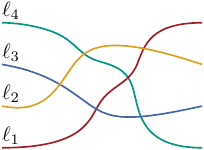}
        \caption{A simple arrangement of four pseudolines.}
        \label{fig:stretchability_arrangement}
    \end{subfigure}
    \hfill
    \begin{subfigure}[t]{0.29\textwidth}
        \centering
        \includegraphics[page=2]{figures/hyperbolic-stretchability.pdf}
        \caption{A corresponding line arrangement in $\R^2$.}
        \label{fig:stretchability_euclidean}
    \end{subfigure}
    \hfill
    \begin{subfigure}[t]{0.32\textwidth}
        \centering
        \includegraphics[page=3]{figures/hyperbolic-stretchability.pdf}
        \caption{A corresponding line arrangement in~$\H^2$ in the \Poincare disk~$D$.}
        \label{fig:stretchability_hyperbolic}
    \end{subfigure}
    \caption{Stretchability of a simple pseudoline arrangement in~$\R^2$ and~$\H^2$.}
    \label{fig:stretchability}
\end{figure}

We say that a pseudoline arrangement~$\calA$ is \emph{stretchable (in~$\R^2$)} if there is a line arrangement that is homeomorphic to~$\calA$, i.e., having exactly the same intersection pattern.
Figure~\ref{fig:stretchability_euclidean} shows one way to stretch the pseudolines from Figure~\ref{fig:stretchability_arrangement}.
Given a combinatorial description~$\calD$ of a (simple) pseudoline arrangement~$\calA$, we denote by (\textsc{Simple})\Stretchability the decision problem whether~$\calA$ is stretchable.
It is well-known that both problems are \ER-complete~\cite{Mnev1988_UniversalityTheorem,RichterGebert2002_UniversalityTheorem,Shor1991_Stretchability} and many \ER-hardness results are by a reduction from one of them~\cite{Bienstock1991_CrossingNumber,Hoffmann2017_PlanarSlopeNumber,Kratochvil1994_SegmentIntersectionGraphs,Schaefer2010_GeometryTopology,Schaefer2013_GraphRealizability}.

Instead of asking for stretchability in the Euclidean plane~$\R^2$, one may also consider stretchability in the hyperbolic plane~$\H^2$.
For example, Figure~\ref{fig:stretchability_hyperbolic} shows a hyperbolic line arrangement in the \Poincare disk model with the same intersection pattern as the pseudolines from Figure~\ref{fig:stretchability_arrangement}.
Recently, Bieker, Bl{\"{a}}sius, Dohse and Jungeblut observed that being a yes- or no-instance of \SimpleStretchability is independent of the underlying plane being~$\R^2$ or~$\H^2$~\cite{Bieker2023_HyperbolicER}, thereby allowing us to use the term \enquote{stretchable} without specifying whether we consider~$\R^2$ or~$\H^2$:

\begin{theorem}[\cite{Bieker2023_HyperbolicER}]
    \label{thm:hyperbolic_simple_stretchability}
    Let~$\mathcal{D}$ be the combinatorial description of a simple pseudoline arrangement.
    Then~$\mathcal{D}$ is stretchable in~$\R^2$ if and only if~$\mathcal{D}$ is stretchable in~$\H^2$.
\end{theorem}

\subsection{Circle Geometry}

A \emph{circle inversion} with respect to a circle~$c$ with midpoint~$m$ and radius~$r$ swaps the interior and exterior of~$c$.
Each point~$p \in \R^2 \setminus \{m\}$ is mapped to another point~$p' \in \R^2 \setminus \{m\}$ such that both lie on the same ray originating from~$m$ and such that $\mathrm{d}(p, m) \cdot \mathrm{d}(p', m) = r^2$ (here~$\mathrm{d}(\cdot,\cdot)$ is the Euclidean distance).
By adding a single \emph{point at infinity} (denoted by~$\infty$) to~$\R^2$ we obtain the so-called \emph{extended plane}, allowing us to extend the definition of a circle inversion to~$\R \cup \{\infty\}$.
Now~$m$ is mapped to~$\infty$ and vice versa.
See Figure~\ref{fig:circle_inversion} for an example.

Circle inversions map circles and straight lines to other circles and straight lines.
Further, they are conformal, i.e., they preserve the angles between crossing lines and circles.
In particular, they map Lombardi drawings to other Lombardi drawings.
See~\cite{Schwerdtfeger1979_CircleGeometry} for a more thorough introduction.

\section{Complexity of Lombardi Drawability}

We prove that deciding whether a graph admits a Lombardi drawing respecting a fixed rotation system is \ER-complete by providing a polynomial-time many-one reduction from \SimpleStretchability.

Our reduction is split into two parts:
We start by transforming a combinatorial description~$\calD$ of a simple pseudoline arrangement into a graph~$G$ with rotation system~$\calR$.
The reduction is such that~$\calD$ is stretchable if and only if~$G$ admits a Lombardi drawing~$\Gamma$ respecting~$\calR$ under the additional restriction that certain cycles in~$G$ must be drawn as circles in~$\Gamma$.
Only then we extend our construction so to enforce the additional restrictions \enquote{automatically} in each Lombardi drawing.

\subsection{Restricted Lombardi Drawings}

The following construction is illustrated in Figure~\ref{fig:restricted_construction}.
Let~$\calD$ be a combinatorial description of a simple arrangement~$\calA = \{\ell_1, \ldots, \ell_n\}$ of~$n \geq 2$ pseudolines, i.e., an instance of the \ER-complete \SimpleStretchability problem.
To recall, this means that for each pseudoline~$\ell_i \in \calA$ we have an ordered list containing the intersections with other pseudolines.
As~$\calA$ is simple, for each pseudoline this list contains exactly~$n - 1$ intersections, one for each other pseudoline.

The first step of the construction is to extend the pseudoline arrangement~$\calA$ by a simple closed curve~$\gamma$ intersecting every pseudoline in~$\calD$ exactly twice, such that~$\gamma$ contains all intersections of~$\calA$ in its interior, see Figure~\ref{fig:restricted_construction_arrangement}.
Let us denote the resulting arrangement by~$\calA_\gamma$.
The combinatorial description~$\calD_\gamma$ of~$\calA_\gamma$ can be obtained by adding for each pseudoline~$\ell_i \in \calA$ (for~$i \in \{1, \ldots, n\}$) one intersection with~$\gamma$ to the beginning and to the end of its list of intersections.
Further, $\calD_\gamma$ contains a list of intersections for~$\gamma$ containing every pseudoline in~$\calA$ exactly twice and whose cyclic ordering is~$\ell_1, \ldots, \ell_n, \ell_1, \ldots, \ell_n$.

Now let~$G_\gamma$ be the following graph:
We start by adding two vertices~$v_i^l$ and~$v_i^r$ per pseudoline~$\ell_i$ corresponding to the left- and rightmost intersections of~$\ell_i$ (these are the ones with~$\gamma$).
These~$2n$ vertices are then connected to a cycle in the order they appear on~$\gamma$.
Next, we connect each pair~$v_i^l$ and~$v_i^r$ by an edge~$e_i$.
The rotation system~$\calR_\gamma$ shall be such that all edges~$e_i$ are on the same side of~$C_\gamma$.
Now for each pseudoline~$\ell_i$ we add a path~$P_i$ from~$v_i^l$ to~$v_i^r$ by iterating through the list of its intersections from left to right.
For each intersection with another pseudoline~$\ell_j$ we add (in this order) two new vertices~$v_{i,j}^l$ and~$v_{i,j}^r$ to the path.
In~$\calR_\gamma$, path~$P_i$ and edge~$e_i$ should be on opposite sides of ~$C_\gamma$.
Let us denote by~$C_i$ the cycle formed by concatenating~$P_i$ with~$e_i$.
Lastly, for each intersection of two pseudolines~$\ell_i$ and~$\ell_j$ with~$i < j$ we connect (in this order)~$v_{i,j}^l$, $v_{j,i}^l$, $v_{i,j}^r$ and~$v_{j,i}^r$ into a $4$-cycle~$C_{i,j}$.
In~$\calR_\gamma$ the circular ordering around each of the four vertices should contain alternately an edge of~$C_{i,j}$ and an edge of~$C_i$ respectively~$C_j$.
See Figure~\ref{fig:restricted_construction_graph} for the complete construction.

\begin{figure}[tb]
    \centering
    \begin{subfigure}[t]{0.45\textwidth}
        \centering
        \includegraphics[page=1]{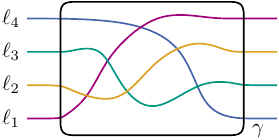}
        \caption{An arrangement of four pseudolines enclosed by a curve~$\gamma$.}
        \label{fig:restricted_construction_arrangement}
    \end{subfigure}
    \hfill
    \begin{subfigure}[t]{0.5\textwidth}
        \centering
        \includegraphics[page=2]{figures/construction.pdf}
        \caption{Graph~$G_\gamma$ drawn such that it respects rotation system~$\calR_\gamma$.}
        \label{fig:restricted_construction_graph}
    \end{subfigure}
    \caption{Example construction of~$G_\gamma$ and~$\calR_\gamma$ from a pseudoline arrangement~$\calA$.}
    \label{fig:restricted_construction}
\end{figure}

\medskip
In the two lemmas below we restrict ourselves to drawings of~$G_\gamma$ in which some cycles must be drawn as circles.
A cycle~$C$ is said to be \emph{drawn as a circle}~$c$ if all vertices and edges of~$C$ lie on~$c$ and the drawing is non-degenerate\footnote{
    A drawing is \emph{degenerate} if two vertices are drawn at the same point or a vertex is drawn in the interior of an edge.
}.
In particular this fixes the ordering of the vertices and edges of~$C$ along~$c$ (the only degree of freedom is whether this ordering is clockwise or counterclockwise).

\begin{lemma}
    \label{lem:stretchable_to_restricted_lombardi}
    If~$\calD$ is stretchable, then~$G_\gamma$ has a Lombardi drawing~$\Gamma$ respecting~$\calR_\gamma$ such that the cycles~$C_\gamma$, all~$C_i$ (for $i \in \{1, \ldots, n\}$) and all~$C_{i,j}$ (for~$i, j \in \{1, \ldots, n\}$ with~$i < j$) are drawn as circles in~$\Gamma$.
\end{lemma}

\begin{proof}
    By Theorem~\ref{thm:hyperbolic_simple_stretchability} we can obtain a hyperbolic line arrangement realizing~$\calD$ in the \Poincare disk model.
    Further, each pseudoline~$\ell_i$ is drawn as a circular arc~$a_i$ (with underlying circle~$c_i$) inside and orthogonal to the \Poincare disk.
    From this, we construct a Lombardi drawing~$\Gamma$ of~$G_\gamma$ respecting~$\calR_\gamma$.

    We denote by~$c_\gamma$ the circle representing the \Poincare disk and draw all vertices of cycle~$C_\gamma$ on it, such that~$v_i^l$ and~$v_i^r$ are placed at the left and right intersection of~$a_i$ with~$c_\gamma$.
    Next, we draw the edges~$e_i$ outside of~$c_\gamma$ on~$c_i \setminus a_i$ and the paths~$P_i$ inside~$c_\gamma$ on~$a_i$.
    As~$v_i^l$ and~$v_i^r$ have degree four and~$c_i$ is orthogonal to~$c_\gamma$, all vertices on~$C_\gamma$ have perfect angular resolution.

    Next, for each pair of intersecting pseudolines~$\ell_i$ and~$\ell_j$ (with $i < j$) we place the vertices~$v_{i,j}^l$ and~$v_{i,j}^r$ to the left, respectively to the right, of the intersections on~$a_i$ (and similar~$v_{j,i}^l$ and~$v_{j,i}^r$ on~$a_j$) such that they lie on a common circle~$c_{i,j}$ which is orthogonal to~$a_i$ and~$a_j$.
    Here, the orthogonality of~$c_{i,j}$ with~$a_i$ and~$a_j$ guarantees perfect angular resolution at the four involved vertices.
    (We prove in Lemma~\ref{lem:circle_around_intersection} in the appendix that such a circle indeed exists).
    Further, we can choose~$c_{i,j}$ small enough so that no two such circles intersect, touch or contain each other.

    The drawing is non-degenerate, respects~$\calR_\gamma$, has~$C_\gamma$, all~$C_i$ and all~$C_{i,j}$ drawn as circles and perfect angular resolution, i.e., it is a Lombardi drawing.
\end{proof}

See Figure~\ref{fig:restricted_lombardi_drawing} for a Lombardi drawing of the graph shown in Figure~\ref{fig:restricted_construction_graph} that is constructed as described in the proof of Lemma~\ref{lem:stretchable_to_restricted_lombardi}.

\begin{figure}[tb]
    \centering
    \includegraphics[page=3]{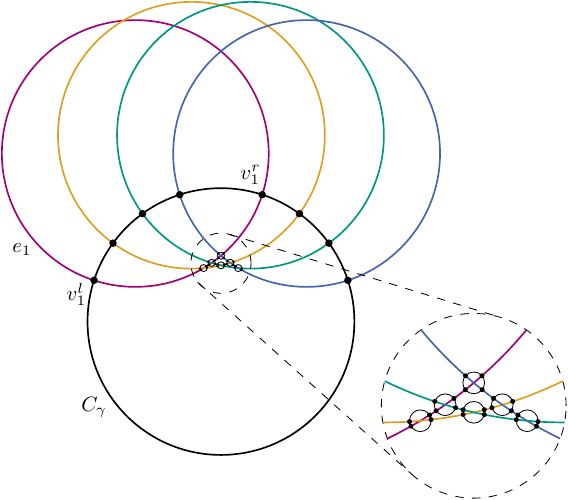}
    \caption{A Lombardi drawing of the graph constructed in Figure~\ref{fig:restricted_construction}.}
    \label{fig:restricted_lombardi_drawing}
\end{figure}

\begin{lemma}
    \label{lem:restricted_lombardi_to_stretchable}
    If~$G_\gamma$ has a Lombardi drawing~$\Gamma$ respecting~$\calR_\gamma$ such that~$C_\gamma$, all~$C_i$ (for $i \in \{1, \ldots, n\}$) and all~$C_{i,j}$ (for~$i, j \in \{1, \ldots, n\}$ with~$i < j$) are drawn as circles in~$\Gamma$, then~$\calD$ is stretchable.
\end{lemma}

\begin{proof}
    Let~$c_\gamma$ be the circle that~$C_\gamma$ is drawn as.
    We can assume without loss of generality that all edges~$e_i$ are drawn outside of~$c_\gamma$ and all paths~$P_i$ are drawn inside of~$c_\gamma$, as we can otherwise consider the drawing obtained from~$\Gamma$ by a circle inversion with respect to~$c_\gamma$.
    Recall that all vertices on~$C_\gamma$ have degree four.
    Let~$c_i$ be the circles that the cycles~$C_i$ are drawn as (for $i \in \{1, \ldots, n\}$).
    All~$c_i$ are orthogonal to~$c_\gamma$ because~$\Gamma$ has perfect angular resolution.

    This allows us to interpret~$c_\gamma$ as a \Poincare disk and each circular arc~$a_i$ of~$c_i$ containing the drawing of~$P_i$ in~$\Gamma$ as a hyperbolic line.
    Therefore, the interior of~$c_\gamma$ induces a hyperbolic line arrangement.
    We prove that this hyperbolic line arrangement has combinatorial description~$\calD$.

    To this end consider an arbitrary but fixed path~$P_i$.
    Along~$P_i$ from left to right we encounter pairs of vertices~$v_{i,j}^l$ and~$v_{i,j}^r$ that, together with~$v_{j,i}^l$ and~$v_{j,i}^r$ form a~$4$-cycle corresponding to the intersection between pseudolines~$\ell_i$ and~$\ell_j$ (assume $i < j$).
    As~$c_i$ and~$c_j$ are circles orthogonal to~$c_\gamma$, their circular arcs~$a_i$ and~$a_j$ intersect at most once.
    Further, as the vertices on~$C_{i,j}$ are alternately on~$P_i$ and~$P_j$, there must be an odd number of intersections between~$a_i$ and~$a_j$ inside the drawing of~$C_{i,j}$ in~$\Gamma$.
    It follows, that~$a_i$ and~$a_j$ intersect exactly once and they do so between~$v_{i,j}^l$ and~$v_{i,j}^r$.
    Thus, each pseudoline intersects each other pseudoline exactly once and in the order described by~$\calD$, i.e., our Lombardi drawing~$\Gamma$ induces a hyperbolic line arrangement with combinatorial description~$\calD$.
    Because~$\calD$ is simple and by Theorem~\ref{thm:hyperbolic_simple_stretchability} it follows that~$\calD$ is then also stretchable in~$\R^2$.
\end{proof}

Summarizing the results so far, we see that Lemmas~\ref{lem:stretchable_to_restricted_lombardi} and~\ref{lem:restricted_lombardi_to_stretchable} give a reduction from \SimpleStretchability to a restricted form of Lombardi drawing.
In what follows, we see how to incorporate these restrictions into the reduction itself.

\subsection{Enforcing the Circles}

In Lemmas~\ref{lem:stretchable_to_restricted_lombardi} and~\ref{lem:restricted_lombardi_to_stretchable} above we assumed that certain cycles in~$G_\gamma$ are drawn as circles.
Below we describe how we can omit this explicit restriction by enforcing all possible Lombardi drawings to \enquote{automatically} satisfy it.

By an \emph{arc-polygon} we denote a set of points~$v_0, \ldots, v_k$ such that~$v_i$ and~$v_{i+1}$ (with~$v_{k+1} = v_0$) are connected by a circular arc or line segment.
An arc-polygon is simple if it does not self-touch or self-intersect.
In case of two or three vertices we speak of a \emph{bigon} and an \emph{arc-triangle}, respectively.
We utilize a lemma by Eppstein, Frishberg and Osegueda~\cite{Eppstein2023_Angles} who characterized simple arc-triangles.
We follow their notation:
The vertices~$v_0$, $v_1$ and~$v_2$ are numbered in clockwise order such that the interior of the arc-triangle is to the right when going from~$v_i$ to~$v_{(i+1) \mod 3}$.
The vertices enclose internal angles~$\theta_0$, $\theta_1$ and~$\theta_2$.
If the vertices do not lie on a common line, then they define a unique circle~$c$.
In this case we denote by~$\phi_i$ the internal angle of the bigon enclosed by~$c$ and the circular arc~$a_i$ between~$v_{(i-1) \mod 3}$ and~$v_{(i+1) \mod 3}$.
Negative (positive) values of~$\phi_i$ mean that~$a_i$ is outside (inside) of~$c$ and~$\phi_i = 0$ means that~$a_i$ is on~$c$.
See Figure~\ref{fig:arc_triangle} for an illustration.

\begin{figure}[tb]
    \begin{minipage}[t]{0.5\textwidth}
        \centering
        \includegraphics{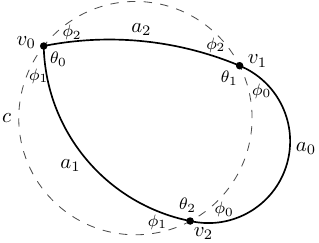}
        \captionof{figure}{A simple arc-triangle illustrating the used notation.}
        \label{fig:arc_triangle}
    \end{minipage}
    \hfill
    \begin{minipage}[t]{0.4\textwidth}
        \centering
        \includegraphics{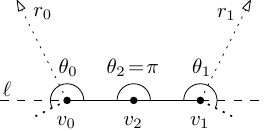}
        \captionof{figure}{
            Ruling out the collinear case in Lemma~\ref{lem:midpoint_on_circle}.
        }
        \label{fig:collinear_case}
    \end{minipage}
\end{figure}

\begin{lemma}[{\cite[Lemma~$4$ and Corollary~$5$]{Eppstein2023_Angles}}]
    \label{lem:bigon_angles}
    Let~$v_0$, $v_1$ and~$v_2$ be a simple arc-triangle as above, not on a common line.
    Then for $\psi = \bigl(\pi - \sum_{i=0}^2 \theta_i\bigr) / 2$ it holds that~$\phi_i = \psi + \theta_i$.
\end{lemma}

We use Lemma~\ref{lem:bigon_angles} to prove that prescribing the interior angles of an arc-triangle to certain values is enough to guarantee that one of its vertices lies on the underlying circle of the circular arc connecting the other two vertices:

\begin{lemma}
    \label{lem:midpoint_on_circle}
    Let~$v_0$, $v_1$ and~$v_2$ be a simple arc-triangle as above such that $\theta_0 = \theta_1 \in [\pi, 3\pi/2)$ and~$\theta_2 = \pi$.
    Further, edge~$v_0v_2$ is drawn as a circular arc~$a_1$ (and not as a line segment) with underlying circle~$c_1$.
    Then~$v_0$, $v_1$ and~$v_2$ do not lie on a common line and~$c_1$ is the unique circle through them, with~$v_1$ on~$c_1 \setminus a_1$.
\end{lemma}

\begin{proof}
    We first rule out the case that all three vertices lie on a common line~$\ell$, see Figure~\ref{fig:collinear_case}:
    As the internal angle at~$v_2$ has size~$\pi$, vertices~$v_0$ and~$v_1$ must be on opposite sides of~$\ell$.
    The internal angle~$\theta_0$ at~$v_0$ defines a ray~$r_0$ that must contain the center of the underlying circle of~$a_2$ (the circular arc connecting~$v_0$ and~$v_1$).
    Similarly,~$\theta_1$ defines another ray~$r_1$ that must contain the center of~$a_2$.
    However, $r_0 \cap r_1 = \emptyset$, because~$\theta_0 = \theta_1 \in [\pi, 3\pi/2)$.
    We conclude that the three vertices cannot lie on a common line and therefore lie on a unique circle~$c$.

    We use Lemma~\ref{lem:bigon_angles} to compute the internal angle of the bigon enclosed by~$a_1$ and~$c$:
    We get that~$\psi = (\pi - \sum_{i=0}^2 \theta_i)/2 = -\theta_1$ and with that $\phi_1 = -\theta_1 + \theta_1 = 0$, i.e.,~$a_1$ must lie on~$c$ and in particular~$c_1 = c$.
    As simple arc-triangles do not self-intersect or self-touch, it follows that~$v_1$ lies on~$c_1 \setminus a$.
\end{proof}

For our reduction we construct, for a given pseudoline arrangement~$\calD$, a graph~$G_\calD$ with rotation system~$\calR_\calD$, which extend~$G_\gamma$ and~$\calR_\gamma$ by new vertices and edges.
As we will see, this forms several arc-triangles in~$G_\calD$ that fulfill the conditions of Lemma~\ref{lem:midpoint_on_circle}.
Iteratively applying this lemma will allow us to prove that~$C_\gamma$ as well as all~$C_i$ (for~$i \in \{1, \ldots, n\}$) and all~$C_{i,j}$ (for $i,j \in \{1, \ldots, n\}$ with $i < j$) must be drawn as circles in any Lombardi drawing of~$G_\calD$ respecting~$\calR_\calD$.

Let us introduce some notation for cycles $C = (v_1, \ldots, v_k)$.
We denote by~$e_i$ the edge of the cycle connecting~$v_i$ and~$v_{i+1}$ (where~$v_{k+1} = v_1$).
If each~$v_i$ (for~$i \in \{1, \ldots, k\}$) has degree four, then the incident edges form four angles between them, which in case of perfect angular resolution must all have size~$\pi/2$.
We call these the \emph{quadrants} of~$v_i$ and label them by~$q_i^1, \ldots, q_i^4$ in counterclockwise order such that~$q_i^1$ is left of~$e_i$ when traversing it from~$v_i$ to~$v_{i+1}$.
By a \emph{half-edge} we denote an incident edge to a vertex whose other endpoint is not yet specified.

A \emph{circle gadget} for a cycle~$C = (v_1, \ldots, v_k)$ as above together with~$k-3$ additional half-edges in each quadrant of all~$v \in V(C)$ is the following set of edges:
For $j \in \{1, \ldots, k-3\}$, the $j$-th half-edge of~$q_1^1$ in clockwise order is joined with the $j$-th half-edge of~$q_{k-j}^2$ in counterclockwise order, see Figure~\ref{fig:circle_gadget}.
The following lemma shows that these edges enforce that~$C$ is drawn as a circle.

\begin{figure}[tb]
    \begin{minipage}[t]{0.63\textwidth}
        \centering
        \includegraphics{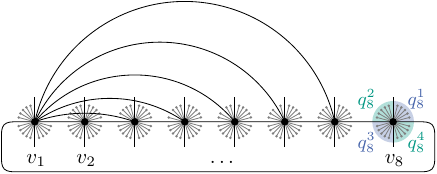}
        \captionof{figure}{
            Circle gadget for a cycle $C = (v_1, \ldots, v_8)$.
            The quadrants at~$v_8$ are labeled.
        }
        \label{fig:circle_gadget}
    \end{minipage}
    \hfill
    \begin{minipage}[t]{0.31\textwidth}
        \centering
        \includegraphics{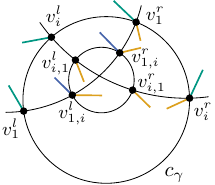}
        \captionof{figure}{
            How to combine different circle gadgets.
        }
        \label{fig:combined_circle_gadgets}
    \end{minipage}
\end{figure}

\begin{lemma}
    \label{lem:circle_forcing}
    Let~$G$ be a graph containing a cycle~$C = (v_1, \ldots, v_k)$ with the edges of a circle gadget as described above.
    Then in every Lombardi drawing~$\Gamma$ that maps edge~$e_k$ to a circular arc~$a$ (i.e., not to a line segment), all vertices and edges of~$C$ are drawn onto the underlying circle~$c$ of~$a$.
\end{lemma}

\begin{proof}
    First note that each vertex has equally many incident (half-)edges in each of its four quadrants, so by the perfect angular resolution of~$\Gamma$, each quadrant spans an angle of~$\pi/2$.
    In particular, between any two consecutive cycle edges there is an angle of~$\pi$ in~$\Gamma$.
    Now consider the three cycle vertices~$v_1$, $v_{k-1}$ and~$v_k$ which form a simple arc-triangle whose internal angles satisfy the conditions of Lemma~\ref{lem:midpoint_on_circle}.
    It follows that~$v_{k-1}$ and~$e_{k-1}$ are drawn onto~$c \setminus a$.

    As we now know that the path from~$v_1$ counterclockwise via~$v_k$ to~$v_{k-1}$ follows a single circular arc~$a'$ in~$\Gamma$, the same argument can be repeated for the simple arc-triangle formed by the points~$v_1$, $v_{k-2}$ and~$v_{k-1}$.
    It follows that~$v_{k-2}$ and~$e_{k-2}$ lie on~$c$.
    Iterating the argument until we reach the simple arc-triangle formed by~$v_1$, $v_2$ and~$v_3$ proves the statement.
\end{proof}

With the circle gadget at hand, we can finally construct~$G_\calD$ and~$\calR_\calD$.
Recall that~$G_\gamma$ is $4$-regular.
We add~$2n-3$ half-edges into each quadrant of every vertex~$v \in V(G_\gamma)$ and then the following circle gadgets:
\begin{itemize}
    \item For the cycle~$C_\gamma = (v_1^l, \ldots, v_n^l, v_1^r, \ldots, v_n^r)$ on~$2n$ vertices.
    Here~$v_1^l$ takes the role of~$v_1$ and~$v_n^r$ takes the role of~$v_k$ in the circle gadget.
    \item For each cycle~$C_i$ on~$2n$ vertices (for~$i \in \{1, \ldots, n\}$).
    Here~$v_i^r$ takes the role of~$v_1$ and~$v_i^l$ takes the role of~$v_k$ in the circle gadget.
    \item For each cycle~$C_{i,j}$ (for~$i, j \in \{1, \ldots, n\}$ with~$i < j$).
    Here~$v_1 = v_{i,j}^l$, $v_2 = v_{j,i}^l$, $v_3 = v_{i,j}^r$ and~$v_4 = v_{j,i}^r$.
\end{itemize}
Note that several vertices are involved in multiple circle gadgets in our construction.
This is not a problem because we carefully placed the circle gadgets such that no two circle gadgets operate in the same quadrant on each vertex.
See Figure~\ref{fig:combined_circle_gadgets} for a visualization (for simplicity, just~$C_1$ and one~$C_i$ are drawn):
Green half-edges show which quadrants are used by the circle gadget for~$C_\gamma$.
Orange half-edges belong to the circle gadgets of~$C_1$ and~$C_i$.
Lastly, blue half-edges belong to the circle gadget of~$C_{1,i}$.

\medskip
As the last step of the reduction, all remaining half-edges are terminated with a new vertex of degree one.
The resulting graph and rotation system are~$G_\calD$ and~$\calR_\calD$.

\begin{lemma}
    \label{lem:stretchable_to_lombardi}
    If~$\calD$ is stretchable, then~$G_\calD$ has a Lombardi drawing respecting~$\calR_\calD$.
\end{lemma}

\begin{proof}
    We start by applying Lemma~\ref{lem:stretchable_to_restricted_lombardi} to obtain a Lombardi drawing~$\Gamma_\gamma$ of the subgraph~$G_\gamma$ of~$G_\calD$ respecting~$\calR_\gamma$ and in which cycle~$C_\gamma$, all~$C_i$ (for~$i \in \{1, \ldots, n\})$ and all~$C_{i,j}$ (for $i,j \in \{1, \ldots, n\}$ with $i < j$) are drawn as circles.

    It remains to draw the vertices and edges added by the circle gadgets.
    Recall that equally many edges were added into each quadrant of all vertices~$v \in V(G_\gamma)$.
    Thus, the angles between edges in~$E(G_\gamma)$ remain unchanged and the new edges must be drawn with equal angles between them into their quadrants to obtain perfect angular resolution.

    Let~$e$ be an edge of a circle gadget with endpoints~$u$ and~$v$ and let~$c$ be the circle that~$u$ and~$v$ lie on in~$\Gamma_\gamma$.
    By construction, $e$ was obtained by joining the $j$-th half-edge in clockwise order in the first quadrant of~$u$ with the $j$-th half-edge in counterclockwise order in the second quadrant of~$v$ for some~$j$.
    Thus the two angles between~$c$ and the arc representing~$e$ at~$u$ and~$v$ are equal and in~$[0,\pi/2)$.
    There is exactly one circular arc that~$e$ can be drawn onto~\cite[Property~$1$]{Duncan2012_Lombardi}.

    Lastly, we need to make sure that the vertices of degree~$1$ that resulted from unjoined half-edges are drawn such that they do not lie on any other edge.
    This can be achieved by drawing the half-edges sufficiently short.
\end{proof}

\begin{lemma}
    \label{lem:lombardi_to_stretchable}
    If~$G_\calD$ has a Lombardi drawing~$\Gamma$ respecting~$\calR_\calD$, then~$\calD$ is stretchable.
\end{lemma}

\begin{proof}
    Recall that edges of~$G_\calD$ are mapped to circular arcs or line segments in~$\Gamma$ and that each vertex has perfect angular resolution.
    We begin by analyzing how the cycle~$C_\gamma = (v_1^l, \ldots, v_n^l, v_1^r, \ldots, v_n^r)$ in~$G_\calD$ must be drawn in~$\Gamma$.
    We can assume that~$v_1^l$, $v_n^r$ and~$v_{n-1}^r$ do not lie on a common line and that the edge between~$v_1^l$ and~$v_n^r$ is mapped to a circular arc~$a$ (by a suitable circle inversion).
    Then by Lemma~\ref{lem:circle_forcing} all vertices and edges of~$C_\gamma$ lie on the underlying circle~$c_\gamma$ of~$a$.

    Next, we consider the circles~$C_i$ for~$i \in \{1, \ldots, n\}$.
    By applying a suitable circle inversion with respect to~$c_\gamma$ if necessary, we can assume that the paths~$P_i$ are drawn inside~$c_\gamma$.
    In converse, the edges~$e_i$ are drawn outside of~$c_\gamma$ and therefore must be drawn as circular arcs (because a line segment would be inside~$c_\gamma$).
    Applying Lemma~\ref{lem:circle_forcing} to each~$C_i$ yields that it is drawn as a circle~$c_i$ in~$\Gamma$.

    It remains to consider the circles~$C_{i,j}$ for~$i,j \in \{1, \ldots, n\}$ with~$i < j$.
    If they are not already drawn as a circle~$c_{i,j}$ we can replace their drawing by a circle with sufficiently small radius by Lemma~\ref{lem:circle_around_intersection} (in the appendix).

    Now it follows from Lemma~\ref{lem:restricted_lombardi_to_stretchable} that~$\calD$ is stretchable.
\end{proof}

\medskip
At this point we can finally prove our main result, Theorem~\ref{thm:lombardi_er_complete}:

\begin{proof}[of Theorem~\ref{thm:lombardi_er_complete}]
    Let~$\calD$ be a combinatorial description of a simple pseudoline arrangement.
    Construct~$G_\calD$ and~$\calR_\calD$ as described above.
    By Lemmas~\ref{lem:stretchable_to_lombardi} and~\ref{lem:lombardi_to_stretchable}, $\calD$ is stretchable if and only if~$G_\calD$ admits a Lombardi drawing respecting~$\calR_\calD$, proving \ER-hardness.
    We prove \ER-membership in Lemma~\ref{lem:membership} (in the appendix).
\end{proof}


\section{Conclusion and Open Problems}

In this paper we proved that it is \ER-complete to decide whether a given graph~$G$ with a fixed rotation system~$\calR$ admits a Lombardi drawing respecting~$\calR$.
To the best of our knowledge, this is the first result on the complexity of Lombardi drawing for general graphs.

In fact, Lombardi drawing is just a special case of the more general problem where instead of enforcing perfect angular resolution we want to draw a graph with circular arc edges and angles of prescribed size.
Our reduction immediately proves \ER-hardness for this problem as well:

\begin{corollary}
    Let~$G$ be a graph with a rotation system~$\calR$ and let~$\Theta$ be an angle assignment prescribing the size of all angles in~$\calR$.
    Then it is \ER-complete to decide whether~$G$ admits a drawing with edges as circular arcs or line segments respecting~$\calR$ and~$\Theta$.
\end{corollary}

On the other hand, several interesting questions remain open:
Our reduction heavily relies on fixing the rotation system~$\calR$.
By the perfect angular resolution requirement this fixes all angles in every Lombardi drawing.
We wonder whether the problem remains \ER-complete without fixing~$\calR$:

\begin{openproblem}
    What is the computational complexity of deciding whether a graph admits any Lombardi drawing (without fixing a rotation system~$\calR$)?
\end{openproblem}

Given a planar graph, one usually asks for a planar Lombardi drawing.
The graphs constructed in our reduction are in general not planar.
In fact, they contain arbitrarily large clique minors.
This motivates our second open problem:

\begin{openproblem}
    What is the complexity of deciding whether a planar graph admits a planar Lombardi drawing (with or without fixing a rotation system~$\calR$)?
\end{openproblem}

\section*{Acknowledgements}

We thank Torsten Ueckerdt, Laura Merker and three anonymous reviewers for carefully reading this manuscript and providing valuable feedback.

\renewcommand\UrlFont{\color{blue}\rmfamily}
\bibliographystyle{splncs04}
\bibliography{references}

\begin{thebibliography}{10}
\providecommand{\url}[1]{\texttt{#1}}
\providecommand{\urlprefix}{URL }
\providecommand{\doi}[1]{https://doi.org/#1}

\bibitem{Bieker2023_HyperbolicER}
Bieker, N., Bl{\"{a}}sius, T., Dohse, E., Jungeblut, P.: {Recognizing Unit Disk
  Graphs in Hyperbolic Geometry is $\exists\mathbb{R}$-Complete}. In:
  Proceedings of the 39th European Workshop on Computational Geometry (EuroCG
  2023). pp. 35:1--35:8 (2023). \doi{10.48550/arXiv.2301.05550}

\bibitem{Bienstock1991_CrossingNumber}
Bienstock, D.: {Some Provably Hard Crossing Number Problems}. Discrete \&
  Computational Geometry  \textbf{6}(3),  443--459 (1991).
  \doi{10.1007/BF02574701}

\bibitem{Canny1988_PSPACE}
Canny, J.: {Some Algebraic and Geometric Computations in PSPACE}. In:
  {Proceedings of the Twentieth Annual ACM Symposium on Theory of Computing}.
  pp. 460--467. STOC '88, Association for Computing Machinery, New York, NY,
  USA (Jan 1988). \doi{10.1145/62212.62257}

\bibitem{Cardinal2015_SimultaneousEmbedding}
Cardinal, J., Kusters, V.: {The Complexity of Simultaneous Geometric Graph
  Embedding}. Journal of Graph Algorithms and Applications  \textbf{19}(1),
  259--272 (2015). \doi{10.7155/jgaa.00356}

\bibitem{Chernobelskiy2012_ForceLombardi}
Chernobelskiy, R., Cunningham, K.I., Goodrich, M.T., Kobourov, S.G., Trott, L.:
  {Force-Directed Lombardi-Style Graph Drawing}. In: van Krefeld, M.,
  Speckmann, B. (eds.) Graph Drawing (GD 2011). Lecture Notes in Computer
  Science, vol.~7034, pp. 320--331 (2012). \doi{10.1007/978-3-642-25878-7_31}

\bibitem{Duncan2018_PlanarPolyArcLombardi}
Duncan, C.A., Eppstein, D., Goodrich, M.T., Kobourov, S.G., L{\"{o}}ffler, M.,
  N{\"{o}}llenburg, M.: {Planar and poly-arc Lombardi drawings}. Journal of
  Computational Geometry  \textbf{9}(1),  328--355 (2018).
  \doi{10.20382/jocg.v9i1a11}

\bibitem{Duncan2012_Lombardi}
Duncan, C.A., Eppstein, D., Goodrich, M.T., Kobourov, S.G., N{\"{o}}llenburg,
  M.: {Lombardi Drawings of Graphs}. Journal of Graph Algorithms and
  Applications  \textbf{16}(1),  85--108 (2012). \doi{10.7155/jgaa.00251}

\bibitem{Duncan2013_Trees}
Duncan, C.A., Eppstein, D., Goodrich, M.T., Kobourov, S.G., N{\"{o}}llenburg,
  M.: {Drawing Trees with Perfect Angular Resolution and Polynomial Area}.
  Discrete \& Computational Geometry  \textbf{49}(2),  157--182 (2013).
  \doi{10.1007/s00454-012-9472-y}

\bibitem{Eppstein2014_SoapBubbles}
Eppstein, D.: {A M{\"{o}}bius-Invariant Power Diagram and Its Applications to
  Soap Bubbles and Planar Lombardi Drawing}. Discrete \& Computational Geometry
   \textbf{52}(3),  515--550 (2014). \doi{10.1007/s00454-014-9627-0}

\bibitem{Eppstein2016_HalinRecognition}
Eppstein, D.: {Simple Recognition of Halin Graphs and Their Generalizations}.
  Journal of Graph Algorithms and Applications  \textbf{20}(2),  323--346
  (2016). \doi{10.7155/jgaa.00395}

\bibitem{Eppstein2021_BipartiteSP}
Eppstein, D.: {Bipartite and Series-Parallel Graphs Without Planar Lombardi
  Drawings}. Journal of Graph Algorithms and Applications  \textbf{25}(1),
  549--562 (2021). \doi{10.7155/jgaa.00571}

\bibitem{Eppstein2021_HyperbolicGD}
Eppstein, D.: {Limitations on Realistic Hyperbolic Graph Drawing}. In:
  Purchase, H.C., Rutter, I. (eds.) Graph Drawing and Network Visualization (GD
  2021). Lecture Notes in Computer Science, vol. 12868, pp. 343--357 (2021).
  \doi{10.1007/978-3-030-92931-2_25}

\bibitem{Eppstein2023_Angles}
Eppstein, D., Frishberg, D., Osegueda, M.C.: {Angles of Arc-Polygons and
  Lombardi Drawings of Cacti}. Computational Geometry  \textbf{112} (2023).
  \doi{10.1016/j.comgeo.2023.101982}

\bibitem{Erickson2022_SmoothingTheGap}
Erickson, J., van~der Hoog, I., Miltzow, T.: {Smoothing the Gap Between NP and
  ER}. SIAM Journal on Computing pp. FOCS20--102--FOCS20--138 (2022).
  \doi{10.1137/20M1385287}

\bibitem{Hilbert1968_GrundlagenGeometrie}
Hilbert, D.: {Grundlagen der Geometrie}. Teubner, 13 edn. (1968).
  \doi{10.1007/978-3-322-92726-2}

\bibitem{Hobbs2003_GlobalNetworks}
Hobbs, R.: {Mark Lombardi: Global Networks}. Independent Curators International
  (2003)

\bibitem{Hoffmann2017_PlanarSlopeNumber}
Hoffmann, U.: {On the Complexity of the Planar Slope Number Problem}. Journal
  of Graph Algorithms and Applications  \textbf{21}(2),  183--193 (2017).
  \doi{10.7155/jgaa.00411}

\bibitem{Kindermann2019_Knots}
Kindermann, P., Kobourov, S.G., L{\"{o}}ffler, M., N{\"{o}}llenburg, M.,
  Schulz, A., Vogtenhuber, B.: {Lombardi Drawings of Knots and Links}. Journal
  of Computational Geometry  \textbf{10}(1),  444--476 (2018).
  \doi{10.20382/jocg.v10i1a15}

\bibitem{Kratochvil1994_SegmentIntersectionGraphs}
Kratochv{\'{i}}l, J., Matou{\v{s}}ek, J.: {Intersection Graphs of Segments}.
  Journal of Combinatorial Theory, Series B  \textbf{62}(2),  289--315 (1994).
  \doi{10.1006/jctb.1994.1071}

\bibitem{Kyncl2011_SimultaneousEmbedding}
Kyn{\v{c}}l, J.: {Simple Realizability of Complete Abstract Topological Graphs
  in $P$}. Discrete \& Computational Geometry  \textbf{45}(3),  383--399
  (2011). \doi{10.1007/s00454-010-9320-x}

\bibitem{Lubiw2022_DrawingInPolygonialRegion}
Lubiw, A., Miltzow, T., Mondal, D.: {The Complexity of Drawing a Graph in a
  Polygonal Region}. Journal of Graph Algorithms and Applications
  \textbf{26}(4),  421--446 (2022). \doi{10.7155/jgaa.00602}

\bibitem{Mnev1988_UniversalityTheorem}
Mn{\"{e}}v, N.E.: {The Universality Theorems on the Classification Problem of
  Configuration Varieties and Convex Polytopes Varieties}. In: Viro, O.Y.,
  Vershik, A.M. (eds.) Topology and Geometry — Rohlin Seminar, Lecture Notes
  in Mathematics, vol.~1346, pp. 527--543. Springer, Berlin, Heidelberg (1988).
  \doi{10.1007/BFb0082792}

\bibitem{Okamoto2019_SegmentNumber}
Okamoto, Y., Ravsky, A., Wolff, A.: {Variants of the Segment Number of a
  Graph}. In: Archambault, D., T{\'{o}}th, C.D. (eds.) Graph Drawing and
  Network Visualization (GD 2019). Lecture Notes in Computer Science, vol.
  11904, pp. 430--443 (2019). \doi{10.1007/978-3-030-35802-0_33}

\bibitem{Purchase2013_UserStudy}
Purchase, H., Hamer, J., N{\"{o}}llenburg, M., Kobourov, S.G.: {On the
  Usability of Lombardi Graph Drawings}. In: Didimo, W., Patrignani, M. (eds.)
  Graph Drawing (GD 2012). Lecture Notes in Computer Science, vol.~7704, pp.
  451--462 (2013). \doi{10.1007/978-3-642-36763-2_40}

\bibitem{RichterGebert2002_UniversalityTheorem}
Richter-Gebert, J.: {Mn{\"{e}}v’s Universality Theorem revisited} (2002),
  \url{https://geo.ma.tum.de/_Resources/Persistent/3/e/a/2/3ea2ad59228a1a24a67d1e994fa77266a599e73a/15_MnevsUniversalityhTheorem.pdf}

\bibitem{Schaefer2010_GeometryTopology}
Schaefer, M.: {Complexity of Some Geometric and Topological Problems}. In:
  Eppstein, D., Gansner, E.R. (eds.) Graph Drawing (GD 2009). Lecture Notes in
  Computer Science, vol.~5849, pp. 334--344 (2010).
  \doi{10.1007/978-3-642-11805-0_32}

\bibitem{Schaefer2013_GraphRealizability}
Schaefer, M.: {Realizability of Graphs and Linkages}. In: Pach, J. (ed.) Thirty
  Essays on Geometric Graph Theory, pp. 461--482. Springer (2013).
  \doi{10.1007/978-1-4614-0110-0_24}

\bibitem{Schaefer2021_KPlanarityFixed}
Schaefer, M.: {Complexity of Geometric k-Planarity for Fixed k}. Journal of
  Graph Algorithms and Applications  \textbf{25}(1),  29--41 (2021).
  \doi{10.7155/jgaa.00548}

\bibitem{Schaefer2021_FixedK}
Schaefer, M.: {On the Complexity of Some Geometric Problems With Fixed
  Parameters}. Journal of Graph Algorithms and Applications  \textbf{25}(1),
  195--218 (2021). \doi{10.7155/jgaa.00557}

\bibitem{Schaefer2021_RAC}
Schaefer, M.: {RAC-Drawability is $\exists\mathbb{R}$-complete}. In: Rutter,
  I., Purchase, H. (eds.) Graph Drawing and Network Visualization (GD 2021).
  Lecture Notes in Computer Science, vol. 12868, pp. 72--86 (2021).
  \doi{10.1007/978-3-030-92931-2_5}

\bibitem{Schaefer2017_FixedPointsNash}
Schaefer, M., {\v{S}}tefankovi{\v{c}}, D.: {Fixed Points, Nash Equilibria, and
  the Existential Theory of the Reals}. Theory of Computing Systems
  \textbf{60},  172--193 (2017). \doi{10.1007/s00224-015-9662-0}

\bibitem{Schwerdtfeger1979_CircleGeometry}
Schwerdtfeger, H.: {Geometry of Complex Numbers}. Dover Publications, Inc.
  (1979)

\bibitem{Shor1991_Stretchability}
Shor, P.W.: {Stretchability of Pseudolines is {NP}-Hard.} In: Gritzmann, P.,
  Sturmfels, B. (eds.) {Applied Geometry And Discrete Mathematics, Proceedings
  of a {DIMACS} Workshop, Providence, Rhode Island, USA, September 18, 1990}.
  {DIMACS} Series in Discrete Mathematics and Theoretical Computer Science,
  vol.~4, pp. 531--554 (1991). \doi{10.1090/dimacs/004/41}

\end{thebibliography}

\appendix

\section{Omitted Details}
\label{sec:omitted_details}

In the proof of Lemma~\ref{lem:stretchable_to_restricted_lombardi} we construct a Lombardi drawing~$\Gamma_\gamma$ of~$G_\gamma$.
There, for each pseudoline~$\ell_i$ the corresponding path~$P_i$ is drawn along a single circular arc.
Given a pair~$\ell_i, \ell_j$ of two pseudolines (with $i < j$), their corresponding circular arcs~$a_i$ and~$a_j$ cross and we draw a sufficiently small circle~$c_{i,j}$ enclosing their intersection.
Then the vertices~$v_{i,j}^l$, $v_{j,i}^l$, $v_{i,j}^r$ and~$v_{j,i}^r$ are placed on the intersection of~$c_{i,j}$ with~$a_i$ and~$a_j$.
To achieve perfect angular resolution, $c_{i,j}$ must be orthogonal to both circular arcs.
See again Figure~\ref{fig:restricted_lombardi_drawing}.
The following lemma guarantees the existence of such a circle:

\begin{lemma}
    \label{lem:circle_around_intersection}
    Let~$a_1$ and~$a_2$ be two circular arcs with a unique proper intersection~$p$ (i.e., not a touching).
    There is a sufficiently small circle~$c$ orthogonal to~$a_1$ and~$a_2$ enclosing~$p$.
\end{lemma}

\begin{proof}
    For~$i \in \{1,2\}$ let~$c_i$ be the underlying circle of~$a_i$ with center~$(x_i, y_i)$ and radius~$r_i$.
    Without loss of generality, we assume that~$r_1 \geq r_2$.
    We denote by~$d$ the Euclidean distance between the centers of~$c_1$ and~$c_2$.
    By the assumption that~$a_1$ and~$a_2$ have a proper intersection, $c_1$ and~$c_2$ must have two intersections, which is the case if and only if~$d < r_1 + r_2$ and~$d > r_1 - r_2$.

    We shall find a circle~$c$ orthogonal to~$c_1$ and~$c_2$ with a tiny but fixed radius~$r$ (whose exact value is to be determined later).
    Circle~$c$ is orthogonal to~$c_i$ if and only if~$r_i^2 + r^2 = d_i^2$, where~$d_i$ denotes the distance between the center of~$c_i$ and~$c$~\cite{Schwerdtfeger1979_CircleGeometry}.
    Thus, the center of~$c$ must be at distance~$d_i = \sqrt{r_i^2 + r^2}$ from the center of~$c_i$.
    In particular, for both~$i \in \{1,2\}$ the center of~$c$ must be on the circle~$c_i'$ with center~$(x_i,y_i)$ and radius~$d_i$.

    It remains to choose~$r$ such that~$c_1'$ and~$c_2'$ have non-empty intersection.
    This is the case if and only if~$d \leq d_1 + d_2$ and~$d \geq \abs{d_1 - d_2} = d_1 - d_2$ (where the last step follows from~$r_1 \geq r_2$).
    The first inequality holds for any choice of~$r > 0$ because $d < r_1 + r_2 < d_1 + d_2$ (here the first inequality follows from~$c_1$ having two intersections with~$c_2$).
    For the second inequality we know that~$d > r_1 - r_2$ (again, because~$c_1$ and~$c_2$ intersect twice), so there is an~$\eps > 0$ such that~$d = r_1 - r_2 + \eps$.
    Any~$r > 0$ such that~$\abs{d_i - r_i} < \eps/2$ works for us:
    Then
    \[
        d
        = r_1 - r_2 + \eps
        > \left(d_1 - \frac{\eps}{2}\right) - \left(d_2 + \frac{\eps}{2}\right) + \eps
        = d_1 - d_2
    \]
    as desired.
    Thus, for sufficiently small~$r > 0$, there is exactly two possible centers for circles that are orthogonal to~$c_1$ and~$c_2$ and enclose one of their intersections each.
    Choose the one corresponding to the intersection between~$a_1$ and~$a_2$.
\end{proof}

\section{\ER-Membership}

\begin{lemma}
    \label{lem:membership}
    Given a graph~$G$ and a rotation system~$\calR$, deciding whether~$G$ admits a Lombardi drawing respecting~$\calR$ is in~\ER.
\end{lemma}

\begin{proof}
    By a result from Erickson, van der Hoog and Miltzow, \ER-membership follows from the existence of a polynomial-time verification algorithm for a real RAM machine\footnote{
        The real RAM extends the classical word RAM by additional registers that contain real numbers (with arbitrary precision).
        The basic arithmetic operations $+$, $-$, $\cdot$, $/$ and even~$\sqrt{\cdot}$ are supported in constant time.
        See~\cite{Erickson2022_SmoothingTheGap} for a formal definition.
    }~\cite{Erickson2022_SmoothingTheGap}.

    Let~$G$ with rotation system~$\calR$ be a yes-instance of the Lombardi drawing problem.
    Then the obvious witness is a Lombardi drawing~$\Gamma$ mapping each vertex to a point and each edge to a circular arc or line segment.
    Given real-valued coordinates for each vertex and a description of each edge as either a circular arc or a line segment, we have to check the following:
    \begin{itemize}
        \item No edge contains a vertex other than its two endpoints.
        \item No two edges share more than one point.
        \item No two vertices are mapped to the same point.
        \item Each vertex has perfect angular resolution.
        \item The rotation system of~$\Gamma$ is~$\calR$.
    \end{itemize}
    \ER-membership follows because all of the above checks can easily be done in polynomial time on a real RAM machine.
\end{proof}

\end{document}